\documentclass[onefignum,onetabnum]{siamonline220329}

\usepackage{lipsum}
\usepackage{amsfonts}
\usepackage{amsfonts, amsmath, bbold, amssymb}
\usepackage{graphicx}
\usepackage{epstopdf}
\usepackage{algorithmic}
\ifpdf
  \DeclareGraphicsExtensions{.eps,.pdf,.png,.jpg}
\else
  \DeclareGraphicsExtensions{.eps}
\fi

\newcommand{\out}[2]{\mathbb{1}_{\mathcal{X}}\left( \{#1\},#2\right)}
\newcommand{\outbar}[2]{\mathbb{1}_{\bar{\mathcal{X}}}\left( \{#1\},#2\right)}

\usepackage{enumitem}
\setlist[enumerate]{leftmargin=.5in}
\setlist[itemize]{leftmargin=.5in}

\newsiamremark{remark}{Remark}
\newsiamremark{example}{Example}
\newsiamremark{hypothesis}{Hypothesis}
\crefname{hypothesis}{Hypothesis}{Hypotheses}
\newsiamthm{claim}{Claim}

\headers{Mathematical Foundations of Data Cohesion}{Katherine E. Moore}

\title{Mathematical Foundations of Data Cohesion}

\author{Katherine E. Moore\thanks{Department of Mathematics and Statistics, Amherst College, Amherst, MA
  (\email{kmoore@amherst.edu}).}}

\usepackage{amsopn}

\ifpdf
\hypersetup{
  pdftitle={Mathematical Foundations of Data Cohesion},
  pdfauthor={K. E. Moore}
}
\fi

\begin{document}

\maketitle

\begin{abstract}  Data cohesion, recently introduced in \cite{pald22} and inspired by social interactions, uses distance comparisons to assess relative proximity.  In this work, we provide a collection of results which can guide the development of cohesion-based methods in exploratory data analysis and human-aided computation.  Here, we observe the important role of highly clustered ``point-like" sets and the ways in which cohesion allows such sets to take on qualities of a single weighted point.  In doing so, we see how cohesion complements metric-adjacent measures of dissimilarity and responds to local density.  We conclude by proving that cohesion is the unique function with (i) average value equal to one-half and (ii) the property that the influence of an outlier is proportional to its mass.  Properties of cohesion are illustrated with examples throughout.  
\end{abstract}



\section{Introduction}

Data cohesion, recently introduced in \cite{pald22}, is a measure of relative proximity inspired by human-social interactions.  In that initial work, from input distance information, cohesion was used to define weighted networks which reveal structural information and, with a simple threshold for distinguishing strong and weak ties, obtain community clusters.  It is also observed in plots of cohesion against distance that cohesion transforms distances in a way that allows one to detect analogous structure occurring in regions with differing  local density (i.e., average within-region distance). Notably, cohesion does not require the use of localizing parameters nor distributional assumptions.

In this work, we establish many properties of cohesion that can help guide the development of cohesion-based methods in exploratory data analysis and human-aided computation.  Specifically, we begin by quickly observing that cohesion considers only the information obtained from dissimilarity comparisons of the form $d(x, y) <_? \text{min}\{d(x, z), d(y, z)\}$.  Then, as a consequence of working with such comparisons alone, cohesion permits sets that are highly concentrated or ``compact" to take on qualities of a single weighted point.  Throughout, we see that cohesion provides an alternative measure of relative proximity with behavior quite different from metric-adjacent measures of distance (or dissimilarity). Connections with the problem of clustering and related work in topological data analysis suggest the value of a measure with these properties.

In this paper, we introduce ``point-like" partitions to provide a new prototypical example for clustering that permits varying average within-cluster distance (or local density) and cluster size.  Then, by generalizing cohesion to take weighted responses to triplet comparisons, and making a slight modification to the original definition, we are able to give straightforward statements of the properties that follow.  For instance, we show that the values of cohesion are constant between distinct point-like sets; this is one aspect of the property that point-like sets take on qualities of a single weighted point.  We then observe that, when an outlier is added to the set, the cohesion among the non-outlier points increases by the weight of the outlier; we say ``the influence of an outlier is proportional to its mass." In particular, unlike metric-adjacent measures of dissimilarity, the cohesion between points is influenced by the mass of surrounding points.  We conclude by proving that cohesion is the unique function with the property that (i) the average value is always equal to one-half and (ii) the influence of an outlier is proportional to its mass.  Examples throughout guide intuition for these and other properties of cohesion.

In \Cref{sec:dissimilarity} we introduce the concept of ``point-like" sets and draw a connection with the property of consistency of a clustering algorithm.  Then, in \Cref{sec:tripletspace}, we define triplet comparison spaces to provide a more general framework that permits weighted responses to dissimilarity comparison queries.  
In \Cref{sec:cohesion}, we define the cohesion function in the more general setting and establish several basic properties.  In  \Cref{sec:quotient}, by considering an appropriate quotient space, we obtain results which highlight how point-like sets take on properties of a single weighted point and the manner in which cohesion accounts for varying density.  In \Cref{sec:uniqueness}, we consider the influence of outliers and prove the uniqueness result for cohesion.   Throughout, we see the ways that cohesion provides structural information that complements that provided by metric-adjacent measures of dissimilarity.  The results provided here can help facilitate the development of cohesion-based methods that leverage information provided by this new measure.

\section{Background and Notation}\label{sec:background}

Throughout, when we consider dissimilarity spaces, \\ $(\mathcal{X}, d, p)$,  we will suppose that  $\mathcal{X} = \{x_1, x_2, \dots, x_n\}$ is a finite set; $d: \mathcal{X}^2 \rightarrow \mathbb{R}^{\ge 0}$ is a symmetric dissimilarity (or distance) function on $\mathcal{X}$ which satisfies $d(x, x) < d(x, y)$ for all $x,y\in \mathcal{X}$ with $y \neq x$; and $p$ is a probability mass function on $\mathcal{X}$.  When $p$ is omitted, we use the uniform distribution, $p = \frac{1}{n}$.   Throughout, given $S \subseteq \mathcal{X}$, we write $m(S) = \sum_{x \in S} p_x$ to denote the probability mass (or weight) of $S$.   For clarity of exposition, we will assume that distinct elements are never exactly the same distance from any other element, that is, $d(x, y) = d(x, z)$ if and only if $y= z$.    Uncertainty in the evaluation of dissimilarity comparisons, including pairs at equal distances, can be handled in a principled way using the perspective introduced in \Cref{sec:tripletspace} (see \Cref{def:Ind_Xd}).  Notably, as we will not directly access the scale of dissimilarity, there is no metric requirement on $d$.  Further, it is only in establishing a uniqueness result in \Cref{thm:uniqueness} that we will use that, in $(\mathcal{X}, d)$, the elements of $\mathcal{X}$ can be ordered according to their dissimilarity from a fixed $x \in \mathcal{X}$.

In the case of $(\mathcal{X}, d)$ as above, the cohesion function, introduced in \cite{pald22}, is as follows.  First, 
for $x, y \in \mathcal{X}$ with $y \neq x$, the $(x, y)$-\textit{local set}, denoted $\mathcal{U}_{x, y}$, (see  \Cref{fig:pointlikeh}) is defined by $$\mathcal{U}_{x, y} = \{ z \in \mathcal{X} |  \min\{d(x, z), d(y, z)\} < d(x, y)  \}.$$
Now, for fixed $x, w \in \mathcal{X}$,  select $Y \in \mathcal{X} \setminus \{x\}$ uniformly at random and then select $Z \in \mathcal{U}_{x,y}$ uniformly at random.  Then the cohesion of $w$ to $x$, denoted $C_{x, w}$ is
$$C_{x, w} = \mathbf{P}(Z = w \text{ and } d(Z, x) < d(Z, Y)) = \frac{1}{n-1} \sum_{y \in \mathcal{X} \setminus \{x\}} \frac{\mathbf{1}(d(w, x) < d(w, y) \text{ and } w \in \mathcal{U}_{x, y} )}{\# \mathcal{U}_{x, y}},$$
where $\mathbf{1}(S)$ is the indicator function for the statement $S$.    The cohesion function is not symmetric in its arguments. 

Cohesion arose from a measure of local (community) depth, also introduced in \cite{pald22}, and given by $\ell(x) = \sum_{w \in \mathcal{X}} C_{x, w}$.  That is, cohesion is \textit{partitioned local depth} and their associated algorithm is referred to as PaLD.   The theoretical results presented in the remainder of this paper can directly provide insight into this related measure of local depth.  

We will first consider the setting of distance comparisons before generalizing to triplet comparison spaces.  When we return to cohesion in \Cref{sec:cohesion}, we will give a definition of cohesion this more general setting and slightly modify the original definition to allow us to write uncomplicated statements of its properties.

\section{Dissimilarity Comparisons and Point-Like Sets}\label{sec:dissimilarity}

We begin with the input setting of $(\mathcal{X}, d)$ and define the ``point-like" sets that will play an important role throughout.  We also draw connections with the property of consistency of a clustering algorithm with respect to transformations that shrink within-cluster distances \cite{kleinberg2002impossibility}.

For many complex data types, it can be challenging to produce an absolute measure of dissimilarity (or distance) for objects in a collection. The dissimilarity comparison framework used throughout can provide an alternative and less restrictive input type. Dissimilarity comparisons are often collected within triplets of points; the three possible forms are standard queries $d(x, y) <_? d(x, z)$ as in \cite{ali2022efficient, amid2015multiview, haghiri2017comparison, tamuz2011adaptively, ukkonen2017crowdsourced,  van2012stochastic};  central queries $d(x, y) <_? \max\{d(x, z), d(y, z)\}$ as in \cite{ kleindessner2017lens}; and outlier queries $d(x, y) <_? \min\{d(x, z), d(y, z)\}$  as in \cite{heikinheimo2013crowd}.  Queries of the form $d(x, y) <_? d(w, z)$, which allow one to detect some degree of density variation have also been considered in \cite{ghoshdastidar2019foundations, terada2014local}.

Some related work in this setting include a method for learning a distance metric  \cite{schultz2003learning} and kernel function \cite{kleindessner2017kernel, tamuz2011adaptively}; obtaining a low-dimensional Euclidean embedding \cite{van2012stochastic}; measuring centrality and data depth \cite{
heikinheimo2013crowd, kleindessner2017lens, rendsburg2021comparison}; determining near neighbors  \cite{haghiri2017comparison}; and performing hierarchical \cite{ghoshdastidar2019foundations} and correlation clustering \cite{ukkonen2017crowdsourced}.  Many of the above are motivated in part by human-aided computation in which query responses are  crowdsourced.  Methods for filling in missing triplet information are proposed in \cite{ali2022efficient, van2012stochastic}; for incorporating weighted responses in \cite{ kleindessner2017kernel, mojsilovic2019relative}; and for efficiently collecting similarity comparisons in \cite{wilber2014cost}.  Additional related work highlights the utility of rank-based methods for applications in which the similarity information, or even the sets themselves, are obtained from from multiple sources \cite{baron2021partitioned, darling2023rank}. 
 
Our focus will be on outlier-type dissimilarity comparisons because they allow certain sets that are highly concentrated or ``compact" to take on qualities of a single weighted point.  We now introduce the concept of point-like sets, see also \Cref{fig:pointlikeh}.  

\begin{definition}\label{def:pointliked} In the setting of $(\mathcal{X}, d)$, we say that a set $X \subseteq \mathcal{X}$ is \emph{point-like} if for any $x, x' \in X$ and $y, z \in \mathcal{X}$ (not both in $X$), we have 
$d(x, y) < \min\{d(x, z), d(y, z)\}$ if and only if $d(x', y) < \min\{d(x', z), d(y, z)\}$.
 If $\mathbf{X} = \{X_1, X_2, ..., X_N\}$ is a partition of $\mathcal{X}$ such that each $X_i$ is point-like, we say that the partition itself is point-like.   \end{definition}

 In other words, point-like sets are those for which the response to comparison queries (with those outside the set) does not depend on which representative of the set is taken.   Note that $\mathbf{X} = \{\{x_1\}, \{x_2\}, \dots, \{x_n\}\}$ and $\mathbf{X} = \{\mathcal{X}\}$ are (trivial) point-like partitions of $\mathcal{X}$.  To see why, in the setting of triplet comparisons, point-like sets must be defined in terms of outlier-type queries, suppose that $x, x'$ are nearly identical and $y$ is clearly distinct (i.e.,  $d(x, x') \ll d(x, y), d(x'y)$).  In considering the three possible forms of triplet comparisons, we note that the truth values of $d(y, x) <_? d(y, x')$ (standard) and $\max\{d(y, x'), d(x,  x') \} <_? d(y, x)$ (central) depend on the ranking of $d(x, y)$ and $d(x', y)$.  It is only outlier-type queries that do not force a distinction between nearly identical elements.  

\begin{figure}[htbp]
  \centering
  \label{fig:pointlikeh}\includegraphics[scale = .8]{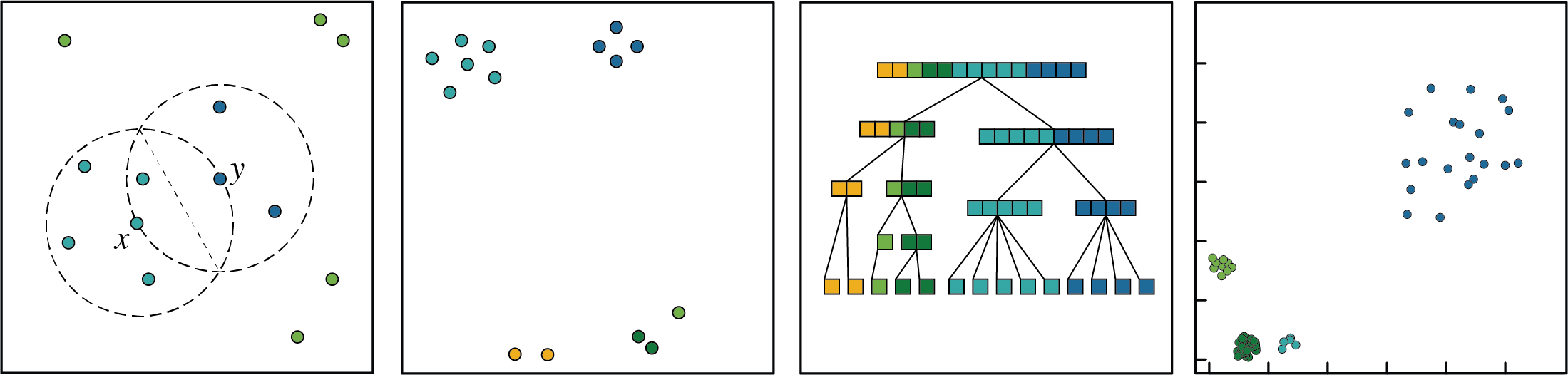}
  \caption{On the right, an $(x, y)$-local set, $U_{x, y}$, for a few points in $\mathbb{R}^2$ with the Euclidean distance; the light blue points are the points, $z$, for which $d(x, z) < \text{min}\{d(x, y), d(y, z)\}$.  Then, a small data set of points with point-like structure at various scales and an illustration of the collection of all point-like sets arranged to highlight the associated partial ordering.  On the right, a set with four point-like sets. Varying density among groups, as in this set, poses challenges for purely distance-based clustering methods such as $k$-means and hierarchical clustering.}
\end{figure}

As seen in \Cref{fig:pointlikeh}, a given space can have point-like structure at varying scales.  For instance, that the number of clusters of a given set may be ambiguous can be argued by arranging points in such a way that there are multiple non-trivial point-like partitions.  The collection of all point-like subsets imposes a hierarchical structure on $\mathcal{X}$ in the following sense.  

\begin{proposition}\label{prop:nested} Given $(\mathcal{X}, d)$, the collection, $\mathcal{C}$, of all point-like sets of $\mathcal{X}$, is partially ordered under subset containment.   That is,  if $X, Y \in \mathcal{C}$,  and $X \cap Y \neq \emptyset$, then either $X \subseteq Y$ or $Y \subseteq X$.  
\end{proposition}

\begin{proof}
 Toward a contradiction, suppose that $X$ and $Y$ are point-like sets satisfying $X \not \subseteq Y$, $Y \not \subseteq X$, and $X \cap Y \neq \emptyset$.    In particular, there are distinct points satisfying $x \in X \setminus Y$, $y \in Y \setminus X$ and $z \in X \cap Y$.  Since $X$ is point-like, $x, z \in X$ and $y \in \mathcal{X} \setminus X$,  it follows that $d(z, x) < d(z, y)$.  On the other hand, since $Y$ is point-like and $y, z \in Y$ and $x \in \mathcal{X} \setminus Y$, we have $d(z, y) < d(z, x)$, thereby contradicting the previous inequality.  
\end{proof}

Point-like partitions are an example of a highly clustered set and many sets with apparent cluster structure will not have a non-trivial point-like partition.  Nevertheless,  sets with non-trivial point-like partitions provide a prototypical example from which to consider properties of clustering algorithms and will play an important role in the remainder of this paper.  

One may hope that a clustering algorithm does not (unnecessarily) divide point-like sets.  The preservation of such sets, however, is not the focus of many traditional clustering algorithms that use absolute (rather than relative) distances throughout.  As one example, since $k$-means minimizes the sum of squared-distances from centroids, the algorithm may break apart spread-out point-like sets, as in \Cref{fig:pointlikeh}. Hierarchical methods and others that do not attempt to account for varying local density can also be challenged by similar arrangements.

The consistency of clusters under certain types of transformations (e.g., those that shrink within-cluster distances) is closely related to an idea considered by Kleinberg \cite{kleinberg2002impossibility}.   Given $(\mathcal{X}, d)$ and partition $\Gamma = \{G_1, G_2, ..., G_N\}$ of $\mathcal{X}$ induced by a given clustering algorithm, one then supposes that the associated dissimilarities are transformed in such a way that (a) $d'(x_i, x'_i) \leq d(x_i, x'_i)$ for any $x_i, x'_i \in G_i$; and (b) $d'(x_i, x_j) \ge d(x_i, x_j)$ for any $x_i \in G_i$ and $x_j \in G_j$ where $i \neq j$.  The property of \textit{consistency} (of the associated clustering algorithm) states that when the clustering algorithm is applied to $(\mathcal{X}, d')$, it must also yield the cluster partition, $\Gamma$.  Kleinberg's impossibility result suggests that consistency may be, in a sense, too strict in that any consistent clustering algorithm cannot also satisfy both scale invariance and richness (i.e., that any partition is achievable by some configuration of points).

The comparison framework gives us a nice way to narrow down the types of clusters and transformations under which we might desire some degree of consistency.  We show below that outlier-type comparisons are exactly those that are preserved under transformations which uniformly shrink distances within point-like sets.

\begin{definition}\label{def:Xtransf} Given $(\mathcal{X}, d)$ and an associated partition $\mathbf{X} = \{X_1, X_2, \dots, X_N\}$,  we say that $d': \mathcal{X}^2 \rightarrow \mathbb{R}$ is obtained from an \emph{$\mathbf{X}$-transformation of $d$} if there are constants $\alpha_i \in (0, 1]$, for $1 \leq i \leq N$, and $\beta \ge 1$ such that 
$$d'(x, y) = \begin{cases}  \alpha_i d(x, y) & \text{ if } x, y \in X_i \\
\beta d(x, y) & \text{ if } x \in X_i \text{~and~} y \in X_j \text{ for some } i \neq j.  \end{cases}$$
\end{definition} 

As we show in the next proposition, algorithms (such as cohesion) built on outlier-type comparisons alone are consistent with respect to $\mathbf{X}$-transformations for any point-like partition, $\mathbf{X}$.

\begin{proposition}\label{prop:shrinking} Suppose $(\mathcal{X}, d)$ and $\textbf{X}$ is an associated point-like partition.  If $d'$ is an $\mathbf{X}$-transformation of $d$, then for any $x, y, z \in \mathcal{X}$, 
\begin{equation} \label{eq:consistent} d(x, y) < \min\{d(x, z), d(y, z) \} \text{ if and only if } d'(x, y) < \min\{d'(x, z), d'(y, z) \}.\end{equation} 
Consequentially, $\mathbf{X}$ is a also point-like partition of $(\mathcal{X}, d')$.  
\end{proposition}

\begin{proof} We consider several cases.  (i) When $x, y, z \in X_i$ for some index $i$, \eqref{eq:consistent} follows immediately from \Cref{def:Xtransf}
as all distances have been scaled by $\alpha_i$.  (ii) Suppose now $x, y \in X_i$ and $z \in X_j$ where $j \neq i$.  By \Cref{def:Xtransf}, $d'(x, y) \leq d(x, y)$ and $d(x, z) \leq d'(x, z)$ and $d(y, z) \leq d'(y, z)$.  Since $X_i$ is point-like,  $d(x, y) < \text{min}\{d(x, z), d(y, z)\}$.  Now, $d'(x, y) \leq d(x, y) <   \min\{d(x, z), d(y, z)\} \leq  \min\{d'(x, z), d'(y, z)\}$, and so  \eqref{eq:consistent} holds. 

(iii) Now suppose $x, z \in X_i$ and $y \in X_j$ for some $j \neq i$.  By \Cref{def:Xtransf}, $d'(x, z) \leq d(x, z)$ and $d'(x, y) \ge d(x, y)$ and $d'(y, z) \ge d(y, z)$.  Since $X_i$ is point-like, $\min\{d(x, z), d(y, z)\} =  d(x, z)$ and $d(x, y) > \text{min}\{d(x, z), d(y, z)\}$. Now, 
$$d'(x, y) \ge d(x, y) >  \text{min}\{d(x, z), d(y, z)\} = d(x, z) \ge d'(x, z) \ge \text{min} \{d'(x, z), d'(y, z)\},$$
 and so \eqref{eq:consistent} follows.  
(iv) The case in which $x \in X_i$ and $y, z \in X_j$ for some $j \neq i$ follows an identical argument to (iii).   Lastly, in the case that (v) $x \in X_i$, $y \in X_j$ and $z \in X_k$ for distinct $i, j, k$, \eqref{eq:consistent} follows immediately from \Cref{def:Xtransf} as all distances have been scaled by $\beta$.  
\end{proof}

Note additionally that, since query responses are invariant under $\mathbf{X}$-transformations, the (hierarchical) point-like structure of $\mathcal{X}$ is also preserved.   As a corollary to \Cref{prop:shrinking}, given a collection of triplet responses, it is not in general possible to reconstruct the complete set of dissimilarity relationships (e.g., $d(w, x) <_? d(y, z)$) among all $w, x, y, z \in \mathcal{X}$; see \cite{terada2014local} for related work including some properties of sets in which recovery is possible.

Therefore, provided that sets are sufficiently separated (i.e., point-like), outlier-type dissimilarity comparisons do not observe the magnitude of within-set distances.  It is this property of outlier-type comparisons that allows cohesion to detect analogous structure regardless of the underlying local density, see \Cref{sec:cohesion}. In the next section, we briefly introduce triplet comparison spaces before defining cohesion in this generalized setting.

\section{Triplet Comparison Spaces}\label{sec:tripletspace}

In this section, we introduce triplet comparison spaces and define point-like sets and $(x, y)$-local sets in this setting.  This more general setting is amenable to human-aided computation and other cases in which one might want to provide weighted responses to similarity comparison queries.    The introduced notation will also allow us to provide quick proofs of many properties of cohesion.  

\begin{definition}\label{def:tripfunction} Given a finite set $\mathcal{X}$, a \emph{triplet comparison function}, $\mathbb{1}_{\mathcal{X}}: \mathcal{X}^3 \rightarrow [0, 1]$ is a function satisfying that, for any $x, y, z \in \mathcal{X}$, 

\begin{enumerate}
\item $\mathbb{1}_{\mathcal{X}}(\{x, y\}, x) = 0$ whenever $y \neq x$,  
\item $\mathbb{1}_{\mathcal{X}}(\{x, y\}, z) = \mathbb{1}_{\mathcal{X}}(\{y, x\}, z)$, 
\item $\mathbb{1}_{\mathcal{X}}(\{x, y\}, z) + \mathbb{1}_{\mathcal{X}}(\{x, z\}, y) + \mathbb{1}_{\mathcal{X}}(\{y, z\}, x) = 1$.  
\end{enumerate}
\end{definition} 

  Note that (1) could have been equivalently stated as $\out{x, x}{y} = 1$ whenever $y \neq x$; and (3) forces  $\out{x, x}{x} = \frac{1}{3}$ for each $x \in \mathcal{X}$. 

\begin{definition} A \emph{triplet comparison space} is a triplet, $(\mathcal{X},  \mathbb{1}_{\mathcal{X}}, p)$, in which \\ $\mathcal{X} = \{x_1, x_2, \dots, x_n\}$ is a finite set, $p: \mathcal{X} \rightarrow [0, 1]$ is a probability mass function, and $\mathbb{1}_{\mathcal{X}}: \mathcal{X}^3 \rightarrow [0, 1]$ is a triplet comparison function.  
\end{definition}

Triplet comparison information can be extracted from a dissimilarity space, obtained with the aid of crowdsourcing, or constructed from available similarity information in some other way.

\begin{definition}\label{def:Ind_Xd}  Suppose that $(\mathcal{X}, d)$ is a dissimilarity space.  The \emph{triplet comparison function induced by $d$} is then $\mathbb{1}_{\mathcal{X}}: \mathcal{X}^3 \rightarrow [0, 1]$ where:
$$\mathbb{1}_{\mathcal{X}}(\{x, y\}, z) =  \begin{cases} 1 & \text{ if } d(x, y) < \emph{min}\{d(x, z), d(y, z)\} \\ 
 \frac{1}{3}  & \text{ if } x = y = z \\ 0 & \text{ otherwise}. \end{cases}$$
\end{definition}

 In the case that $\mathcal{X}$ has pairs of points at equal distances, one can resolve ties probabilistically.  For instance, in the case that $d(x, y) = d(x, z) < d(y, z)$, one could take $\out{x, y}{z} = \out{x, z}{y} = \frac{1}{2}$ and $\out{y, z}{x} = 0$.  

In the setting of human-aided computations, given a set $\mathcal{X}$, suppose that for distinct $x, y, z \in \mathcal{X}$ we have collected responses to queries of the form: Among $x, y$ and $z$, which two are most alike?  Then, for distinct $x, y, z \in \mathcal{X}$, take $\out{x, y}{z}$ to be the proportion of times that, when the triplet $\{x, y, z\}$ was displayed, the response was that $x$ and $y$ were most alike (i.e., $z$ is the outlier).  Then, for distinct $x, y \in \mathcal{X}$, take $\out{x, x}{y} = 1; \out{x, y}{x} = 0$ and $\out{x, x}{x} = \frac{1}{3}$.  For triplets $\{x, y, z\}$ that do not have a response, one could resolve ties probabilistically by setting $\out{x, y}{z} = \out{x, z}{y} = \out{y, z}{x} = \frac{1}{3}$ or fill in missing values using transitivity.  

In some applications, it is more natural to instead collect triplet comparisons of the form $d(x, y) <_? d(x, z)$.  In the case that the triplets $d(x, y) < d(x, z)$, $d(z, x) < d(z, y)$, and $d(y, z) < d(y, x)$  (or the other circular direction) are all considered true, we are then forced to consider the three distances equal and thus and allocate a $1/3$ to each outlier-type comparison.  Such ranking may be said to be disconcordant; see \cite{darling2023rank} for related work on concordant ranking systems.  For ease of notation in the next equation, for weighted responses write $p_{x, y, z} = \mathbb{1}(d(x, y) < d(x, z))$.   
We could then take
$$\out{x, y}{z} = p_{x, y, z}p_{y, x, z} + \frac{1}{3}(p_{x, y, z}p_{y, z, x}p_{z, x, y} + p_{x, z, y}p_{y, x, z}p_{z, y, x}).  $$

Although not considered here, one could also use the relative magnitude of dissimilarities to provide weighted responses to outlier-type queries.   Most results in the remainder of this paper are presented in the general setting of triplet comparison spaces.  We next state the definition of point-like sets in this setting.

\begin{definition}\label{def:pointlike1} Given a triplet comparison space, $(\mathcal{X}, \mathbb{1}_{\mathcal{X}}, p)$, a set $X \subseteq \mathcal{X}$ is \emph{point-like} if for any $x, x' \in X$ and $y, z \in \mathcal{X}$ not both in $X$,  we have $\out{x, y}{z} = \out{x', y}{z}$.
\end{definition}

We again note the immediate consequence that if $X$ is point-like, then for any $x, x' \in X$ and $y \in \mathcal{X} \setminus X$, $\out{x, x'}{y} = 1$ and $\out{x, y}{x'} = 0$. The concept of $(x, y)$-local sets can be extended to the setting of $(\mathcal{X}, \mathbb{1}_{\mathcal{X}}, p)$ provided that we allocate the probability mass according to the weight given to the associated similarity responses.

\begin{definition}\label{def:uxy}
Given a triplet comparison space, $(\mathcal{X}, \mathbb{1}_{\mathcal{X}}, p)$,  and $x, y \in \mathcal{X}$,  define the \emph{local mass function}, $U_{\mathcal{X}}: \mathcal{X}^2 \rightarrow [0, 1]$, by:
$$U_{\mathcal{X}}(x, y) = \sum_{z \in \mathcal{X}} (\out{x, z}{y} + \out{y, z}{x})p_z.$$

\end{definition}
The induced $(x, y)$-local set mass function generalizes from $(\mathcal{X}, d)$  using \Cref{def:Ind_Xd},  in the sense that $U_{\mathcal{X}}(x, y) = m(\mathcal{U}_{x, y})$ provided that $x \neq y$; and note $U_{\mathcal{X}}(x, x) = 2p_x/3$.  In describing the cohesion function in the setting $(\mathcal{X}, \mathbb{1}_{\mathcal{X}}, p)$, we use the phrase  ``$(x, y)$-local set" to refer to the fuzzy set of points, $z\in \mathcal{X}$, with membership function equal to $\out{x, z}{y} + \out{y, z}{x}$.

\section{Cohesion}\label{sec:cohesion} We now provide a definition of cohesion in the setting of similarity comparison spaces and show the close relationship with that originally given in \cite{pald22}.  In this section, we also establish a few basic properties of cohesion and include examples throughout.

\begin{definition}\label{def:cohesion} Given $(\mathcal{X}, \mathbb{1}_{\mathcal{X}}, p)$ as above,  define the \textit{cohesion function}, $C_{\mathcal{X}}: \mathcal{X}^2 \rightarrow \mathbb{R}^{\ge 0}$, by
$$C_{\mathcal{X}}(x, w) = \sum_{y \in \mathcal{X}} \frac{\out{x, w}{y}}{U_{\mathcal{X}}(x, y)}p_y = \sum_{y \in \mathcal{X}} \frac{\out{x, w}{y}}{\sum_{z \in \mathcal{X}} (\out{x, z}{y} +\out{y, z}{x})p_z}p_y.$$
\end{definition}

 Note that, in contrast to that in \cite{pald22} (see \Cref{sec:background}),  the sum is taken over all $y \in \mathcal{X}$, and pairs at equal distances are resolved slightly differently.   In the case of $(\mathcal{X}, d, p)$, where $p$ is uniform, the relationship between the two definitions of cohesion is 
\[ \label{eq:pald_c} C_{\mathcal{X}}(x, w) = (n-1)\left(C_{x, w} + \frac{1}{2n}\mathbf{1}(x = w)\right).\]
The minor modifications of the original definition allow uncomplicated statements of the properties given in the remainder of this paper.

An interpretation of cohesion is as follows; see also \cite{pald22} for discussion of the human-social perspective. Given $x \in \mathcal{X}$, considering an opposing point $y \in \mathcal{X}$, we now restrict the domain to the $(x, y)$-local points, re-scaling the mass of each point by $1/U_{\mathcal{X}}(x, y)$ (so that the total mass of the $(x, y)$-local points is 1).  Each $(x, y)$-local point then proportionally allocates its (magnified) mass to the focal point, $x$ or $y$, that it is more similar to.  Aggregating over all possible $y \in \mathcal{X}$,  $C_{\mathcal{X}}(x, w)p_w$ is then the total mass that has been allocated to $x$ by $w$ in this process; and $C_{\mathcal{X}}(x, w)$ is the factor by which the mass of $w$ has been amplified.   In this way, cohesion can be interpreted as the overall influence of $w$'s affinity (or ``support") for $x$ when compared with a (general) opposing point.  

Cohesion is a measure of relative proximity with the property that the average value of cohesion is always equal to $1/2$.  Also observed in \cite{pald22}, this property follows immediately from the symmetry of $(x, y)$-local sets. We give the proof using \Cref{def:cohesion} for completeness.  

\begin{proposition}\label{prop:half} For any $(\mathcal{X}, \mathbb{1}_{\mathcal{X}}, p)$, the (weighted) average value of cohesion over the set $\mathcal{X}$ is $1/2$.  That is, 
$$\sum_{x, w \in \mathcal{X}} C_{\mathcal{X}}(x, w)p_xp_w = \frac{1}{2}.$$
\end{proposition}

\begin{proof} Write $\mathcal{X} = \{x_1, x_2, \dots, x_n\}$.  As $U_{\mathcal{X}}(x_i, x_j) =U_{\mathcal{X}}(x_j, x_i)$ for any $x_i, x_j \in \mathcal{X}$ and $\out{x_i, x_j}{x_i} = 0$ when $j \neq i$, we have
\begin{eqnarray*}
\sum_{ i, j } C(x_i, x_j)p_{x_i} p_{x_j} &=& 
\sum_{ i, j, k }\frac{ \out{x_i, x_j}{x_k}}{U_{\mathcal{X}}(x_i, x_k)}  p_{x_i} p_{x_j}p_{x_k} \\
 &=& \sum_{ i, k }\frac{ \sum_j \out{x_i, x_j}{x_k} p_{x_j}}{U_{\mathcal{X}}(x_i, x_k)}  p_{x_i} p_{x_k} \\
&=& \sum_{ i < k } \frac{ U_{\mathcal{X}}(x_i, x_k)}{U_{\mathcal{X}}(x_i, x_k)}p_{x_i} p_{x_k} + \sum_{ i  } \frac{\out{x_i, x_i}{x_i}p_{x_i}}{U_{\mathcal{X}}(x_i, x_i)}p_{x_i} \\
&=& \sum_{i< k} p_{x_i} p_{x_k}   + \sum_{ i } \frac{1}{2} p_{x_i} = \frac{1}{2}.
\end{eqnarray*}
\end{proof}

The extension of cohesion to permit weighted responses allows $w$ to proportionally allocate its mass to the points in the associated $(x, y)$-local sets for which it is a member.  

\begin{example} Suppose $\mathcal{X} = \{\blacksquare,  \bullet, \circ\}$ and among $\blacksquare, \bullet, \circ$, the proportion of responses stating $\blacksquare$ and $\bullet$ are most alike is $\alpha$; and the remaining reply that $\bullet$ and $\circ$ are most alike.  Then, we can take $\out{\blacksquare, \bullet}{\circ} = \alpha$ and $\out{\bullet, \circ}{\blacksquare} = 1- \alpha$.  By  \Cref{def:cohesion}, it follows that
$$C(\blacksquare, \bullet) = \alpha/3 \text{ and } C(\circ, \bullet)= (1- \alpha)/3.$$
\end{example}

Points that are closer to one another typically have larger values of cohesion.  In the following example we see, however, that this is not always the case.  Notice also in this example that cohesion is influenced by the mass of the surrounding points and thus, in this way, cohesion behaves quite differently from distance in the usual metric sense.  

\begin{figure}[htbp]
  \centering
  \label{fig:distcohesion}\includegraphics[scale = .85]{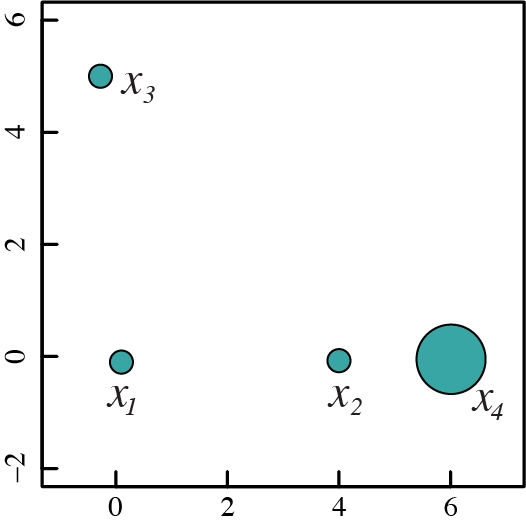}
  \caption{A data set considered in \Cref{ex:ordering} consisting of four weighted points. Provided that the probability mass of $x_4$ is greater than $1/3$, we have $d(x_1, x_2) < d(x_1, x_3)$ and yet $C_{\mathcal{X}}(x_1, x_2) < C_{\mathcal{X}}(x_1, x_3)$.  
As we will see in \Cref{prop:ptlike_calc}, the same phenomena would occur if $x_4$ were replaced with a point-like set.  In contrast to metric-adjacent measures of distance, the cohesion of points are influenced by the weight of others.  
  }
\end{figure}

\begin{example}\label{ex:ordering}  Let $\mathcal{X} = \{x_1, x_2, x_3, x_4\} \in \mathbb{R}^2$ are such that $x_1 = (0, 0), x_2 = (4, 0),  x_3 = (0, 5), x_4 = (6, 0)$; and take $p_{x_4} = p$ and $p_{x_1} = p_{x_2} = p_{x_3} = (1-p)/3$ (see \Cref{fig:distcohesion}).  In particular, such points satisfy $\out{x_1, x_2}{x_3} = \out{x_1, x_3}{x_4} = \out{x_2, x_4}{x_1} = 1$.  Now, since $\out{x_1, x_2}{x_4} = 0$ and $\out{x_1, x_2}{x_1} = \out{x_1, x_2}{x_2} = 0$ (see \Cref{def:tripfunction}), we have

$$C_{\mathcal{X}}(x_1, x_2) =  \frac{\out{x_1, x_2}{x_3}}{U_{\mathcal{X}}(x_1, x_3)}p_{x_3} = \frac{1}{1 - p}\left(\frac{1- p}{3}\right) = \frac{1}{3}.$$
On the other hand, as $\out{x_1, x_3}{x_2} = 0$, and immediately $\out{x_1, x_3}{x_1} = 0 $ and $ \out{x_1, x_3}{x_3} = 0$, we have

$$C_{\mathcal{X}}(x_1, x_3) = \frac{\out{x_1, x_3}{x_4}}{U_{\mathcal{X}}(x_1, x_4)} p_{x_4} = p.$$
Therefore, if we take $p > 1/3$, we have $d(x_1, x_2) < d(x_1, x_3)$ and yet $C_{\mathcal{X}}(x_1, x_2) < C_{\mathcal{X}}(x_1, x_3)$.  This phenomena can also be achieved using the uniform distribution on a set $\mathcal{X}$ provided that we replace $x_4$ with a point-like set with sufficient total probability mass, see \Cref{prop:ptlike_calc} and \Cref{fig:pointlike}.
\end{example}

It is a quick observation that the cohesion of a point to itself is strictly greater than the cohesion from any other point.  That is, for any $x, w \in \mathcal{X}$ with $x \neq w$,  $C_{\mathcal{X}}(x, x) > C_{\mathcal{X}}(x, w)$.  We now show a similar property when $X$ is a point-like set.

\begin{proposition} Given $(\mathcal{X}, \mathbb{1}_{\mathcal{X}}, p)$ and $X \subseteq \mathcal{X}$ is point-like.   Then for any $x, x' \in X$ and $w \in \mathcal{X} \setminus X$ (with  $p_w \neq 0$), $C_{\mathcal{X}}(x, x') > C_{\mathcal{X}}(x, w)$.  
\end{proposition}

\begin{proof} Suppose that $x, x' \in X$ and $w \in \mathcal{X} \setminus X$ (with $p_w \neq 0$).  In the case that $y \in \mathcal{X} \setminus X$, by \Cref{def:pointlike1}, $\out{x, x'}{y} = 1$ and so $\out{x, x'}{y} \ge \out{x, w}{y}$.  On the other hand, in the case that $y \in X$, again by \Cref{def:pointlike1}, we have $\out{x, w}{y} = 0$, and so $\out{x, x'}{y} \ge \out{x, w}{y}$.  Therefore, for any $y \in \mathcal{X}$, $\out{x, x'}{y} - \out{x, w}{y} \ge 0$ and the inequality is strict when $y = w$. 
Hence, 
$$C_{\mathcal{X}}(x, x') - C_{\mathcal{X}}(x, w) = \sum_{y \in X} \frac{\out{x, x'}{y} - \out{x, w}{y}}{U_{\mathcal{X}}(x, y)}p_y  > 0.$$

\end{proof}

There is not a universal upper bound for values of cohesion; if we restrict to the case of the uniform distribution on $\mathcal{X}$, we have the following.  

\begin{example} For a fixed $n \ge 3$ and the uniform probability mass  function $p = \frac{1}{n}$, the maximum values of $C_{\mathcal{X}}(x, w)$ for $x, w \in \mathcal{X}$ are achieved by the configuration $\mathcal{X} = \{x_1, x_2, \dots, x_n\} \subseteq \mathbb{R}$ where $x_i = 1/(2 + \varepsilon)^i$.  In such a case, 
$$C_{\mathcal{X}}(x_n, x_{n-1}) = \sum_{k = 3}^n \frac{1}{k} \text{ and }   C_{\mathcal{X}}(x_n, x_n) = 1 + \sum_{k = 3}^n \frac{1}{k}.$$
\end{example}

Since contracting distances within point-like sets does not change the evaluation of outlier-type dissimilarity comparisons (and thus the values of the induced triplet comparison function),  it follows immediately that cohesion is invariant under such transformations.   

\begin{corollary}[to \Cref{prop:shrinking}] Given $(\mathcal{X}, d, p)$ together with a point-like partition $\mathbf{X}$ of $\mathcal{X}$.  If $(\mathcal{X}, d', p)$ is obtained from a $\mathbf{X}$-transformation (i.e., one which shrinks within-set distances and uniformly expands between-set distances), the values of cohesion over $(\mathcal{X}, d', p)$ are equal to those over $(\mathcal{X}, d, p)$.  
\end{corollary}

To motivate our consideration of quotient spaces, in \Cref{fig:pointlike}, we display the cohesion matrix, $[C_{\mathcal{X}}]_{i, j} = [C_{\mathcal{X}}(x_i, x_j)]$, for a small Euclidean set, $\mathcal{X} \subseteq \mathbb{R}^2$, with a point-like partition consisting of three sets with 20, 30 and 50 points, respectively.  Notice that point-like sets take on qualities of a single weighted point, see \cref{prop:ptlike_calc}, below.   In the following example, we compute the values of cohesion for an associated quotient set, $\bar{\mathcal{X}}$.  

\begin{figure}[htbp]
  \centering
  \label{fig:pointlike}\includegraphics[scale = 1]{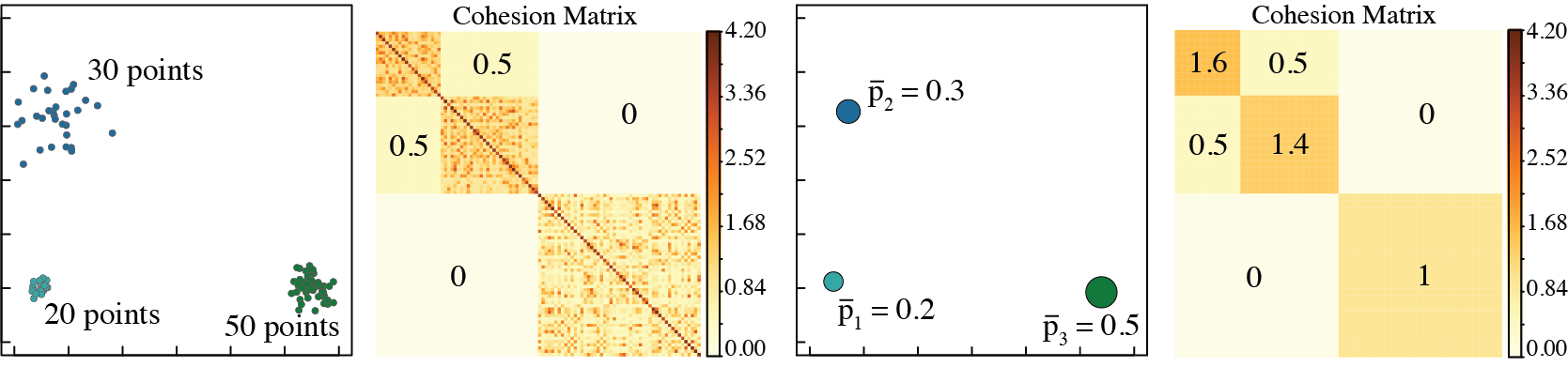}
  \caption{On the left, a set $\mathcal{X} \subseteq \mathbb{R}^2$ consisting of three point-like sets with varying size and density.  The associated cohesion matrix highlights that between group values of cohesion are constant.  On the right, each point-like set has been replaced with a weighted point.  We return to the fractal-like behavior witnessed in the display of the cohesion matrices in \Cref{prop:ptlike_calc}.  }
\end{figure}

\begin{example} Suppose $\bar{\mathcal{X}} = \{\bar{x}_1, \bar{x}_2, \bar{x}_3\}$ where $d(\bar{x}_1, \bar{x}_2) < d(\bar{x}_1, \bar{x}_3) < d(\bar{x}_2, \bar{x}_3)$ and $\bar{p}_{1} = .2, \bar{p}_2 = .3, \bar{p}_3 = .5$. Using \Cref{def:cohesion}, we have
$$C_{\bar{\mathcal{X}}}(\bar{x}_1, \bar{x}_1) = \frac{1/3}{(2/3)(.2)} (.2) + \frac{1}{.5}(.3) + \frac{1}{1}(.5) = 1.6 \text{ and } $$ $$ C_{\bar{\mathcal{X}}}(\bar{x}_1, \bar{x}_2) = \frac{0}{(2/3)(.2)}(.2) + \frac{0}{.5}(.3) + \frac{1}{1}(.5) = 0.5.$$
In the same way, one obtains $C_{\bar{\mathcal{X}}}(\bar{x}_2, \bar{x}_1) = 0.5$, $C(\bar{x}_2, \bar{x}_2) = 1.4$, $C(\bar{x}_3, \bar{x}_3) = 1$, and the remaining values are zero, see \cref{fig:pointlike}.  
\end{example}

\section{Point-Like Partitions and Quotient Spaces}\label{sec:quotient}

In this section, building on \cref{fig:pointlike}, we introduce an associated quotient space which allows us to show how cohesion permits point-like sets to take on qualities of a single weighted point.  As point-like sets can have quite different average within-cluster distances, the properties in this section allow one to see how, when sets are sufficiently separated as to be point-like, cohesion does not witness local density.  Cohesion can then reveal relationships among points that can otherwise be obscured by absolute distance.

\begin{definition}\label{def:quotient} Given $(\mathcal{X}, \mathbb{1}_{\mathcal{X}},p)$ and point-like partition $\textbf{X} = \{X_1, X_2, ..., X_N\}$, let $\bar{\mathcal{X}} = \{\bar{x}_1, \bar{x}_2, ..., \bar{x}_N\}$ where $\bar{x}_i \in X_i$ for $1 \leq i \leq N$ be a set of representatives.   The associated \emph{quotient triplet comparison space} is $(\bar{\mathcal{X}}, \mathbb{1}_{\bar{\mathcal{X}}}, \bar{p})$, where the probability mass function is $\bar{p}_{i} = \sum_{x \in X_i} p_x$ and $\mathbb{1}_{\bar{\mathcal{X}}}: \bar{\mathcal{X}}^3 \rightarrow [0, 1]$ is given by
$\mathbb{1}_{\bar{\mathcal{X}}}(\{\bar{x}_i, \bar{x}_j\}, \bar{x}_k) = \mathbb{1}_{\mathcal{X}}(\{\bar{x}_i, \bar{x}_j\}, \bar{x}_k).$
\end{definition}

Given $X \subseteq \mathcal{X}$ we define an induced triplet comparison subspace; note that $X$ need not be a point-like set of $\mathcal{X}$, though it often will be.

\begin{definition}\label{def:subspace} Given a similarity comparison space, $(\mathcal{X}, \mathbb{1}_{\mathcal{X}}, p)$ and $X \subseteq \mathcal{X}$,  the induced \emph{triplet comparison subspace} is $(X, \mathbb{1}_X, p|^X)$, where the probability mass function is $p|^X_x = \frac{p_x}{m(X)}$ and $\mathbb{1}_{X}: \mathbf{X}^3 \rightarrow [0, 1]$ is given by $\mathbb{1}_X(\{x, y\}, z) = \mathbb{1}_{\mathcal{X}}(\{x, y\}, z)$.  
\end{definition}

\noindent That the induced subspace and quotient spaces are also a triplet comparison space follows immediately from \Cref{def:pointlike1,def:tripfunction}. In \Cref{fig:fourgroups}, we illustrate some of the properties established below.  

\begin{figure}[htbp]
  \centering
  \label{fig:fourgroups}\includegraphics[scale = 1]{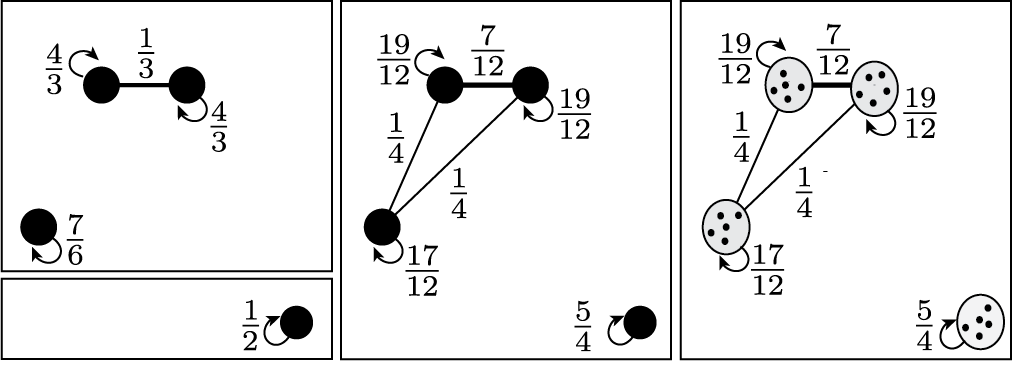}
  \caption{The non-zero values of cohesion are described with edge labels; the indicated values of cohesion are in this case symmetric.  In the center, an outlier is added; cohesion with the outlier is zero, and all other values in the main group have increased by $1/4$ (the mass of the outlier) and the self-cohesion of the outlier is increased by $3/4$.  We refer to this property as ``the influence of an outlier is proportional to its mass", see \Cref{def:linearp}.  On the right, an illustration of point-like behavior considered in \Cref{prop:ptlike_calc}. }
\end{figure}

We will require the following lemma, whose proof is in the Appendix, for the next proposition.

\begin{lemma}\label{lem:uxy} Given $(\mathcal{X}, \mathbb{1}_{\mathcal{X}}, p)$ and point-like partition $\mathbf{X} = \{X_1, X_2, ..., X_N\}$,  let $\bar{\mathcal{X}} = \{\bar{x}_1, \bar{x}_2, ..., \bar{x}_N\}$ be  an associated set of representatives defining the quotient space $(\bar{\mathcal{X}}, \mathbb{1}_{\bar{\mathcal{X}}}, \bar{p})$.  Then for $x \in X_i$ and $y \in X_j$, we have
\[U_{\mathcal{X}}(x, y) = \begin{cases} m(X_i)U_{X_i}(x, y) & \text{ if } i = j \\
 U_{\bar{\mathcal{X}}}(\bar{x}_i, \bar{x}_j) & \text{ if } i \neq j. \end{cases}\]
\end{lemma}

The following proposition describes the way in which point-like sets take on qualities of a single weighted point.  In particular, the values of cohesion are constant between distinct point-like sets.  This property can be observed in \Cref{fig:pointlike}.  As point-like sets can have very different average within-cluster distances, it is this property that explains how cohesion accounts for varying local density.

\begin{proposition}\label{prop:ptlike_calc} Given $\mathcal{X}$ with point-like partition $\mathbf{X} = \{X_1, X_2, ..., X_N\}$,  let \\ $\bar{\mathcal{X}} = \{\bar{x}_1, \bar{x}_2, ..., \bar{x}_N\}$ be  an associated set of representatives.   Then for any $x \in X_i$ and $w \in X_j$, 
\begin{displaymath}C_{\mathcal{X}}(x, w) = \begin{cases} C_{\bar{\mathcal{X}}}(\bar{x}_i, \bar{x}_j) & \text{ if }  i \neq j, \\ C_{\bar{\mathcal{X}}}(\bar{x}_i, \bar{x}_i)  + (C_{X_i}(x, w) - \frac{1}{2}) & \text{ if } i = j.\end{cases}\end{displaymath}
And consequentially, for any $1 \leq i, j \leq N$, 
\begin{displaymath}\sum_{x \in X_i} \sum_{w \in X_j} C_{\mathcal{X}} (x, w)p_x p_w  = C_{\bar{\mathcal{X}}}(\bar{x}_i, \bar{x}_j) p_{\bar{x}_i} p_{\bar{x}_j}.\end{displaymath}
\end{proposition}

\begin{proof}
Suppose that $x \in X_i$ and $w \in X_j$. Using first that when $y \in X_k$ for some $k \neq i$, by \cref{def:pointlike1}, $\mathbb{1}_{\mathcal{X}}(\{x, w\}, y) = \mathbb{1}_{\bar{\mathcal{X}}}(\{\bar{x}_i, \bar{x}_j\}, \bar{x}_k)$ and $U_{\mathcal{X}}(x, y) = U_{\bar{\mathcal{X}}}(\bar{x}_i, \bar{x}_k)$ (see \cref{lem:uxy}), we have
\begin{equation} \label{eq:fractal1} \sum_{y \in \mathcal{X} \setminus X_i} \frac{\out{x, w}{y}}{U_{\mathcal{X}}(x, y)} p_y =  \sum_{k \neq i} \sum_{y \in \mathcal{X}_k}
\frac{\mathbb{1}_{\bar{\mathcal{X}}}(\{\bar{x}_i, \bar{x}_j\},\bar{x}_k)}{U_{\bar{\mathcal{X}}}(\bar{x}_i, \bar{x}_k)}p_y = \sum_{k \neq i} 
\frac{\mathbb{1}_{\bar{\mathcal{X}}}(\{\bar{x}_i, \bar{x}_j\}, \bar{x}_k)}{U_{\bar{\mathcal{X}}}(\bar{x}_i, \bar{x}_k)}p_{\bar{x}_k}.\end{equation}

We first consider the case in which $i = j$. Using \Cref{def:tripfunction}, and that $U_{\bar{\mathcal{X}}}(\bar{x}_i, \bar{x}_i) = 2 \mathbb{1}_{\bar{\mathcal{X}}}(\{\bar{x}_i, \bar{x}_i\}, \bar{x}_i)p_{\bar{x}_i}$, we have
\begin{equation} \label{eq:fractal2}\sum_{k \neq i} 
\frac{\mathbb{1}_{\bar{\mathcal{X}}}(\{\bar{x}_i, \bar{x}_i\}, \bar{x}_k)}{U_{\bar{\mathcal{X}}}(\bar{x}_i, \bar{x}_k)}p_{\bar{x}_k} = C_{\bar{\mathcal{X}}}(\bar{x}_i, \bar{x}_i) - \frac{1}{2}.\end{equation}
Further, when $y \in X_i$, we have $\out{x,w}{y} = \mathbb{1}_{X_i}(\{x, w\}, y)$ and $U_{\mathcal{X}}(x, y) = U_{X_i}(x, y)m(X_i)$ (see \cref{lem:uxy}).  Together with the fact that $p|_y^{X_i} = \frac{p_y}{m(X_i)}$, we have
\begin{equation} \label{eq:fractal3} \sum_{y \in X_i} \frac{\out{x, w}{y}}{U_{\mathcal{X}}(x, y)} p_y = \sum_{y \in X_i} \frac{\mathbb{1}_{X_i}(\{x, w\}, y)}{ U_{X_i}(x, y)m(X_i)} p_y = \sum_{y \in X_i} \frac{\mathbb{1}_{X_i}(\{x, w\}, y)}{U_{X_i}(x, y)} p|^{X_i}_y = C_{X_i}(x, w). \end{equation}

Hence, using \eqref{eq:fractal3} for the first sum and \eqref{eq:fractal1} and \eqref{eq:fractal2}, for the second, we have
\begin{displaymath}C_{\mathcal{X}}(x, w) = \sum_{y \in X_i} \frac{
\out{x, w}{y}}{U_{\mathcal{X}}(x, y)} p_y + \sum_{y \in \mathcal{X} \setminus X_i} \frac{\out{x, w}{y}}{U_{\mathcal{X}}(x, y)} p_y = C_{X_i}(x, w) + C_{\bar{\mathcal{X}}}(\bar{x}_i, \bar{x}_j) - \frac{1}{2}.\end{displaymath}

We now, consider the case in which $i \neq j$. For $y \in X_i$, we have $\out{x, w}{y} = 0$ and $\outbar{\bar{x}_i, \bar{x}_j}{\bar{x}_i} = 0$, and so, with \eqref{eq:fractal1}, 
\begin{displaymath} 
C_{\mathcal{X}}(x, w) = \sum_{k \neq i} \sum_{y \in X_k} \frac{\out{x, w}{y}}{U_{\mathcal{X}}(x, y)}p_y = \sum_{k \neq i} 
\frac{\outbar{\bar{x}_i, \bar{x}_j}{\bar{x}_k}}{U_{\bar{\mathcal{X}}}(\bar{x}_i, \bar{x}_k)}p_{\bar{x}_k} = C_{\bar{\mathcal{X}}}(\bar{x}_i, \bar{x}_j).
\end{displaymath}
\end{proof}

As an example of the application of \cref{prop:ptlike_calc}, we show how the property ``separation under increasing distance" given in \cite{pald22} follows as a corollary.  

\begin{corollary}\label{cor:ptlike_calc} Suppose $\mathcal{X} = A \cup B$ are such that for any $a, a' \in A$ and $b, b' \in B$, \\ $\max\{d(a, a'), d(b, b')\} < d(a, b)$.   Then for any $a \in A$ and $b \in B$, $C_{\mathcal{X}}(a, b) = C_{\mathcal{X}}(b, a) = 0$.
\end{corollary}

\begin{proof} First, notice that $A$ and $B$ are point-like sets.  Fix $\bar{a} \in A$ and $\bar{b} \in B$ and let $\bar{\mathcal{X}} = \{\bar{a}, \bar{b}\}$ be an associated quotient space in which $p_{\bar{a}} = m(A)$ and $p_{\bar{b}} = m(B)$ and $\out{\bar{a}, \bar{b}}{y}  = 0$ for both $y \in \bar{\mathcal{X}}$.   Using \cref{def:cohesion}, $$C_{\bar{\mathcal{X}}}(\bar{a}, \bar{b}) =  \frac{\out{\bar{a}, \bar{b}}{\bar{a}}}{U_{\bar{\mathcal{X}}}(\bar{a}, \bar{a})}p_{\bar{a}} + \frac{\out{\bar{a}, \bar{b}}{\bar{b}}}{U_{\bar{\mathcal{X}}}(\bar{a}, \bar{b})}p_{\bar{b}}  = 0.$$ Now, by \cref{prop:ptlike_calc}, for any $a \in A$ and $b \in B$,  $C_{\mathcal{X}}(a, b) = C_{\bar{\mathcal{X}}}(\bar{a}, \bar{b}) = 0$.  
\end{proof}

Note also that the property ``limiting irrelevance of density" in \cite{pald22} is also a special case of \Cref{prop:ptlike_calc} in which $\mathcal{X}$ is the union of two structurally identical point-like sets.

\section{The Role of Outliers and a Uniqueness Result}\label{sec:uniqueness}

We now turn our attention to the influence of outliers on the cohesion between points.   We conclude with a main result which shows that cohesion is the unique function with constant average value and for which the influence of an outlier is proportional to its mass.

\begin{example} Suppose $(\mathcal{X}, d, p)$ is a dissimilarity space such that $\mathcal{X} = X \cup \{z\}$ where $z$ is an outlier in the sense that $\max_{x, x' \in X} \{d(x, x')\} < \min_{x \in X} \{d(x, z)\}$.  Take $x, w \in X$.  Since $X$ is point-like, for any $y \in X$, $\mathbb{1}_{\mathcal{X}}(\{x, w\}, y) = \mathbb{1}_{X}(\{x, w\}, y)$ and $U_{\mathcal{X}}(x, y) = U_{X}(x, y)m(X)$ (see  \Cref{lem:uxy}). Further, $\mathbb{1}_{\mathcal{X}}(\{x, w\}, z) = 1$ and, using \Cref{def:uxy}, $U_{\mathcal{X}}(x, z) = 1$.    Lastly, for $y \in X$, $p|^X_y = p_y/m(X)$.  We now conclude
\begin{displaymath}C_{\mathcal{X}}(x, w) = \sum_{y \in X}
\frac{\mathbb{1}_{X}(\{x, w\}, y)}{U_{X}(x, y)m(X)}p_y  +\frac{\out{x, w}{z}}{U_{\mathcal{X}}(x, z)}p_z = C_{X}(x, w) +p_z. \end{displaymath}
Lastly, as $\out{x, z}{y} = 0$ for all $y \in \mathcal{X}$, it is immediate that $C_{\mathcal{X}}(x, z)  = C_{\mathcal{X}}(z, x) = 0$. 
\end{example}

That the influence of an outlier is proportional to its mass, considered in the previous example, is also illustrated in \Cref{fig:fourgroups}.  In the case of multiple outliers, we have the following bounds, see also \Cref{fig:outlier}.  

\begin{figure}[htbp]
  \centering
  \label{fig:outlier}\includegraphics[scale = 1]{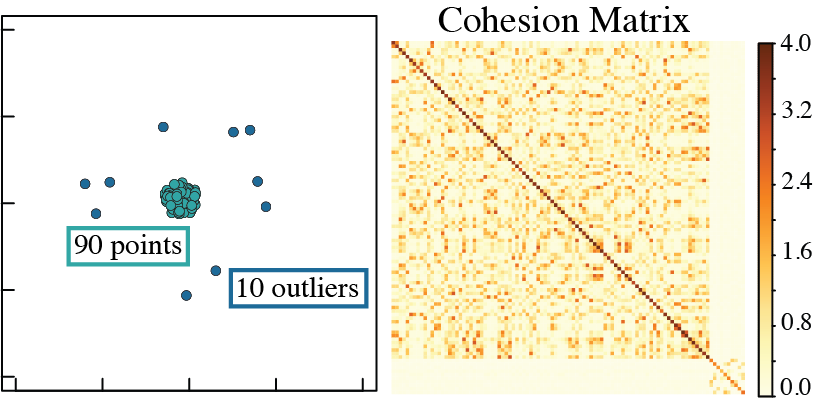}
  \caption{A set $\mathcal{X} \subseteq \mathbb{R}^2$ consisting of points drawn uniformly at random from the unit ball together with ten outliers. Here, the maximum cohesion of an outlier to a non-outlier is $0.093$; the upper bound on this quantity, given by \cref{prop:outliers_bound}, is in this case $0.111$.  }
\end{figure}

\begin{proposition}\label{prop:outliers_bound} Suppose $(\mathcal{X}, d, p)$ is a dissimilarity space and $\mathcal{X} = X \cup Z$ is such that  $\max_{x, x' \in X}\{d(x, x')\} <  \min_{x \in X, z \in Z} \{d(x, z)\}$.  Then, for any $x, w \in X$ and $z \in Z$, 
\begin{displaymath}   C_{\mathcal{X}}(x, z) \leq \frac{p_Z}{1 - p_Z} \text{ and } p_Z \leq C_{\mathcal{X}}(x, w) - C_{X}(x, w)   \leq \frac{p_Z}{1- p_Z}.\end{displaymath}
\end{proposition}

\begin{proof} Suppose first that $x \in X$ and $z \in Z$.   In the case that $y \in Z$,  $\out{x, w}{y} = 1$ for any $w \in X$ and so \begin{equation}\label{eq:prop_uxybound} U_{\mathcal{X}}(x, y)  = \sum_{w \in \mathcal{X}} \left(\out{x, w}{y} + \out{y, w}{x}\right)p_w \ge \sum_{w \in X} \out{x, w}{y}p_w = 1 - p_Z. \end{equation} In the case that $y \in X$, we have $\out{x, z}{y} = 0$.  Together with  \eqref{eq:prop_uxybound}, we have 
\[C_{\mathcal{X}}(x, z) = \sum_{y \in X} \frac{\out{x, z}{y}}{U_{\mathcal{X}}(x, y)}p_y + \sum_{y \in Z} \frac{\out{x, z}{y}}{U_{\mathcal{X}}(x, y)} p_{y} \le \sum_{y \in Z} \frac{1}{1 - p_Z} p_{y} = \frac{p_Z}{1-p_Z}.\]
Now suppose that $x, w \in X$. Since for $y \in X$, we have $U_{X}(x, y)m(X) = U_{\mathcal{X}}(x, y)$ (see \cref{lem:uxy}) and $p|^{X}_y = \frac{p_y}{m(X)}$, we have
\[C_{X}(x, w) =  \sum_{y \in X} \frac{\mathbb{1}_{X}(\{x, w\}, y)}{U_{X}(x, y)} p|^{X}_y = \sum_{y \in X} \frac{\out{x, w}{y}}{U_{\mathcal{X}}(x, y)} p_y.\]
Now, using that $\out{x, w}{y} = 1$ for any $y \in Z$, we have
\begin{equation}\label{eq:prop_bound} C_{\mathcal{X}}(x, w) - C_{X}(x, w) = \sum_{y \in Z} \frac{\out{x, w}{y}}{U_{\mathcal{X}}(x, y)} p_{y} = \sum_{y \in Z} \frac{1}{U_{\mathcal{X}}(x, y)} p_{y}.\end{equation}
Using \eqref{eq:prop_uxybound}, we have $0 \leq \sum_{y \in Z} \frac{1}{U_{\mathcal{X}}(x, y)} p_y \leq \frac{p_Z}{1-p_Z}$, and the result now follows from \eqref{eq:prop_bound}.  
\end{proof}

We now formally define the property observed above that the influence of an outlier is proportional to its mass.  We then show that this property, together with the property of a constant average value, uniquely determine the cohesion function.

\begin{definition}\label{def:linearp} Given a dissimilarity space $(\mathcal{X}, d)$, we say that $x, z \in \mathcal{X}$ are \emph{mutually outlying} if for any $y \in \mathcal{X}$,  $d(x, z) > \min\{d(x, y), d(y, z)\}$.  We then  say that the influence of an outlier is \emph{proportional to its mass} if whenever $x, z \in \mathcal{X}$ are mutually outlying, $g_{\mathcal{X}}(x, z) = 0$ and 
$$g_{\mathcal{X}}(x, w) = \begin{cases} g_{\mathcal{X} \setminus \{z\}}(x, w) + p_z & \text{ if } d(w, x) < d(w, z) \\ 
g_{\mathcal{X} \setminus \{z\}}(x, w)  & \text{ if } d(w, x) > d(w, z).
\end{cases}$$
\end{definition}

We conclude this section with the main result that the cohesion function is the unique function with a constant average value and for which the influence of an outlier is proportional to its mass.

\begin{theorem}\label{thm:uniqueness} Given $(\mathcal{X}, d, p)$, the cohesion function (in Definition~\ref{def:cohesion}) is the unique real-valued function of $\mathcal{X}^2$ for which (i) the weighted average value is equal to $1/2$ and (ii) the influence of an outlier is proportional to its mass.  
\end{theorem}

\begin{proof} Suppose that for any $(\mathcal{X}, d, p)$, the function $g_{\mathcal{X}}: \mathcal{X}^2 \rightarrow \mathbb{R}$ is such that (i) the weighted average value of $g_{\mathcal{X}}$ over $\mathcal{X}^2$ is equal to $1/2$  
and (ii) the influence of an outlier is proportional to its probability mass in the sense of \Cref{def:linearp}.  In particular, in the case of $\mathcal{X}_2 = \{x_1, x_2\}$ and any probability mass function, $p$, the pair $\{x_1, x_2\}$ is mutually outlying and so by (ii), $g(x_1, x_2) = g(x_2, x_1) = 0$.  
Since by (i), that $\sum_{x, w \in \mathcal{X}} g_{\mathcal{X}}(x, w)p_x p_w = \frac{1}{2}$ holds for any probability mass function, if $g_{\mathcal{X}}(x, w)$ did not depend on $p$ it must be that $g_{\mathcal{X}}(x, w) + g_{\mathcal{X}}(w, x) = 1$ for any $x, w \in \mathcal{X}$.  However, this is not the case, since we just observed that for $\mathcal{X}_2$, we have $g(x_1, x_2) + g(x_2, x_1) = 0$.  Let us now write 
\begin{equation}\label{eq:thm_gf} g_{\mathcal{X}}(x, w) = \sum_{y \in \mathcal{X}} f_{\mathcal{X}}(x, w, y) p_y,\end{equation}
 for some unknown function $f_{\mathcal{X}}: \mathcal{X}^3 \rightarrow \mathbb{R}$ of $(\mathcal{X}, d, p)$.  

Given $(\mathcal{X}, d)$ suppose $w_1, z_1, z_2 \in \mathcal{X}$ are such that $\{z_1, z_2\}$ is a mutually outlying pair.  By property (ii), we must have $g_{\mathcal{X}}(z_1, z_2) = \sum_{y \in \mathcal{X}} f_{\mathcal{X}}(z_1, z_2, y)p_y = 0$ for any probability measure, $p$, associated with this $(\mathcal{X}, d)$.  Specifically, the previous statement must hold for the probability mass functions of the form $p^{(x)}_y = \mathbf{1}(y = x)$.  We now conclude $f_{\mathcal{X}}(z_1, z_2, y) = 0 \text{ for each } y \in \mathcal{X}$.

Now for an arbitrary fixed $w_1 \in \mathcal{X}$, we first consider the case in which $d(w_1, z_1) < d(w_1, z_2)$.  Again by property (ii), we must have $g_{\mathcal{X}}(z_1, w_1) = g_{\mathcal{X} \setminus \{z_2\}}(z_1, w_1) + p_{z_2}$.  Since, for a given probability mass function, $p$, the induced probability measure on $\mathcal{X} \setminus \{z_2\}$ is $(p|^{\mathcal{X} \setminus \{z_2\}})_y = \frac{p_y}{1 - p_{z_2}}$, expanding we obtain
\begin{eqnarray*} 0 &=& g_{\mathcal{X}}(z_1, w_1) - (g_{\mathcal{X}\setminus \{z_2\}}(z_1, w_1) + p_{z_2}) \\
&=&  \sum_{y \in \mathcal{X}} f_{\mathcal{X}}(z_1, w_1, y) p_y - \left( \sum_{y \in \mathcal{X} \setminus \{z_2\}} f_{\mathcal{X} \setminus \{z_2\}}(z_1, w_1, y) \frac{p_y}{1 - p_{z_2}}  + p_{z_2} \right) \\
&=& (f_{\mathcal{X}}(z_1, w_1, z_2) - 1)p_{z_2} + \sum_{y \in \mathcal{X} \setminus \{z_2\}} \left(f_{\mathcal{X}}(z_1, w, y) - \frac{f_{\mathcal{X} \setminus \{z_2\}}(z_1, w_1, y)}{1 - p_{z_2}}\right) p_y. 
\end{eqnarray*}
As the previous above expression holds regardless of the probability mass function, $p$, on $(\mathcal{X}, d)$, using again the probability mass functions of the form $p^{(x)}_y = \mathbf{1}(y = x)$,  each of the above coefficients must be zero. Therefore, in the case that $\{z_1, z_2\}$ are mutually outlying and $d(w_1, z_1) < d(w_1, z_2)$, we have $f_{\mathcal{X}}(z_1, w_1, z_2) = 1$ and $f_{\mathcal{X}}(z_1, w_1, y) = \frac{f_{\mathcal{X} \setminus \{z_2\}}(z_1, w_1, y)}{1 - p_{z_2}}$ for any $ y \in \mathcal{X} \setminus \{z_2\}$.  

Lastly, consider the case in which $w_1 \in \mathcal{X}$ satisfies $d(w_1, z_1) > d(w_1, z_2)$.  Again by property (ii), we must have $g_{\mathcal{X}}(z_1, w_1) = g_{\mathcal{X} \setminus \{z_2\}}(z_1, w_1)$.  Expanding, as above, we obtain
\begin{eqnarray*} 0 &=& g_{\mathcal{X}}(z_1, w_1) - g_{\mathcal{X}\setminus \{z_2\}}(z_1, w_1) \\
&=&  \sum_{y \in \mathcal{X}} f_{\mathcal{X}}(z_1, w_1, y) p_y -\sum_{y \in \mathcal{X} \setminus \{z_2\}} f_{\mathcal{X} \setminus \{z_2\}}(z_1, w_1, y) \frac{p_y}{1 - p_{z_2}}  \\
&=& f_{\mathcal{X}}(z_1, w_1, z_2) p_{z_2} + \sum_{y \in \mathcal{X} \setminus \{z_2\}} \left(f_{\mathcal{X}}(z_1, w, y) - \frac{f_{\mathcal{X} \setminus \{z_2\}}(z_1, w_1, y)}{1 - p_{z_2}}\right) p_y. 
\end{eqnarray*}
As the previous expression holds for any $p$ on $(\mathcal{X}, d)$, each of the above coefficients must be zero.  In summary, if $\{z_1, z_2\}$ are mutually outlying $f(z_1, z_2, y) = 0$ for all $y \in \mathcal{X}$ and 
\begin{equation}\label{eq:thm_fxsetminusz_summary} f_{\mathcal{X}}(z_1, w_1, y) = 
\begin{cases}  \frac{f_{\mathcal{X} \setminus \{z_2\}}(z_1, w_1, y)}{1 - p_{z_2}}  & \text{if } y \in \mathcal{X} \setminus \{z_2\}, \\
1 & \text{if }  y = z_2 \text{ and } d(w_1, z_1) < d(w_1, z_2), \\
0 & \text{if } y = z_2 \text{ and } d(w_1, z_1) > d(w_1, z_2). \\
\end{cases}\end{equation}

Consider now an arbitrary $(\mathcal{X}, d)$ and $x \in \mathcal{X}$. Index the elements of $\mathcal{X}$ according to their dissimilarity from $x$, and so $d(x, z_n) > d(x, z_{n-1}) > ... > d(x, x)$ and $z_1 = x$.    For ease of notation, for any index $1 \leq i \leq n$, we write $\mathcal{X}_i = \{z_1, z_2, ..., z_i\}$ and note that the probability mass function associated with $\mathcal{X}_i$ is $(p^{|\mathcal{X}_i})_{z_l} = \frac{p_{z_l}}{m(\mathcal{X}_i)}$ and $m(\mathcal{X}_i) = \sum_{k = 1}^i p_{z_k}$.  

In the case that $j > 1$, we see that, since $x$ and $z_j$ are mutually outlying in the set $\mathcal{X}_j$ as $d(x, z_j) > d(x, y)$ for all $y \in \mathcal{X}_j \setminus \{z_j\}$.  By \eqref{eq:thm_fxsetminusz_summary}, it follows that
\begin{equation}\label{eq:thm_fxxxj} \frac{f_{\mathcal{X}_j}(x, x, x)}{m(\mathcal{X}_j)} =\left(\frac{1}{m(\mathcal{X}_j)}\right)\left(\frac{f_{\mathcal{X}_{j-1}}(x, x, x)}{1 - (p^{|{\mathcal{X}_j}})_{z_j}}\right) =  \frac{f_{\mathcal{X}_{j-1}}(x, x, x)}{m(\mathcal{X}_j)(1 - \frac{p_{z_j}}{m(\mathcal{X}_j)})} = \frac{f_{\mathcal{X}_{j-1}}(x, x, x)}{m(\mathcal{X}_{j-1})}.\end{equation}
Note that for the set $\mathcal{X}_1 = \{x\}$, since by property (i) the average value of $g_{\mathcal{X}_1}$ must be $1/2$,  $f_{\mathcal{X}_1}(x, x, x) = \frac{1}{2}$.  Applying the argument as in \eqref{eq:thm_fxxxj} $n-1$ times, we obtain
\begin{equation}\label{eq:thm_fxxx} f_{\mathcal{X}}(x, x, x) = \frac{f_{\mathcal{X}_n}(x, x, x)}{m(\mathcal{X}_n)} = \frac{f_{\mathcal{X}_1}(x, x, x)}{m(\mathcal{X}_1)} = \frac{1}{2p_x}.\end{equation}

Since by (i), the weighted average value of $g_{\mathcal{X}}$ is equal to $1/2$ and so, with  \eqref{eq:thm_gf}, we have  $\sum_{x, w, y \in \mathcal{X}} f_{\mathcal{X}}(x, w, y) p_x p_w p_y = \frac{1}{2}$. Further, $f_{\mathcal{X}}(x, x, x)p_x = \frac{1}{2}$ for each $x \in \mathcal{X}$.  
As dissimilarity comparisons are independent of the probability mass function and these two expressions hold for any $p$, at least one of the following symmetries must be present: 
\begin{enumerate}
\item For any $x, w \in \mathcal{X}$, we have $\sum_{y \in \mathcal{X}} (f_{\mathcal{X}}(x, w, y) + f_{\mathcal{X}}(w, x, y)) p_y = 1$. 
\item For any $w, y \in \mathcal{X}$, we have $\sum_{x \in \mathcal{X}} (f_{\mathcal{X}}(x, w, y) + f_{\mathcal{X}}(x, y, w))p_x = 1$. 
\item For any $x, y \in \mathcal{X}$, we have $\sum_{w \in \mathcal{X}} (f_{\mathcal{X}}(x, w, y) + f_{\mathcal{X}}(y, w, x)) p_w =1$.
\end{enumerate}
It is clear that it is not the first since, when $x$ and $w$ are mutually outlying, 
\begin{displaymath}\sum_{y \in \mathcal{X}} (f_{\mathcal{X}}(x, w, y) + f_{\mathcal{X}}(w, x, y)) p_y = g_{\mathcal{X}}(x, w) + g_{\mathcal{X}}(w, x) = 0.\end{displaymath}
We show that for $\mathcal{X} = \{0, 4, 6, 11\} \subseteq \mathbb{R}$ that the second statement does not hold for $w = x_2$ and $y = x_4$.  Since $\{x_1, x_4\}$ are mutually outlying, by \eqref{eq:thm_fxsetminusz_summary} $f_{\mathcal{X}}(x_1, x_4, x_2) = 0$, and further since $d(x_2, x_1) < d(x_2, x_4)$ we have $f_{\mathcal{X}}(x_1, x_2, x_4) = 1$.  Similarly, since $\{x_2, x_4\}$ are mutually outlying,  again by \eqref{eq:thm_fxsetminusz_summary}, $f_{\mathcal{X}}(x_2, x_4, x_2)  = f_{\mathcal{X}}(x_4, x_2, x_4)= 0$ and $f_{\mathcal{X}}(x_2, x_2, x_4) = f_{\mathcal{X}}(x_4, x_4, x_2) = 1$.  
Hence,
$f(x, x_2, x_4) + f(x, x_4, x_2) = 1 \text{ for each } x \in \{x_1, x_2, x_4\}$.

Now, observe that $\{x_3, x_4\}$ are mutually outlying in $\mathcal{X} \setminus \{x_1\}$, and so $f_{\mathcal{X} \setminus \{x_1\}}(x_3, x_4, x_2) = 0$.  Further still, since $d(x_2, x_3) < d(x_2, x_4)$ we have $f_{\mathcal{X} \setminus \{x_1\}}(x_3, x_2, x_4) = 1$.  Now, since $\{x_3, x_1\}$ are mutually outlying and both $d(x_2, x_3) < d(x_2, x_1)$ and $d(x_4, x_3) < d(x_4, x_1)$,
leveraging again \eqref{eq:thm_fxsetminusz_summary}, we obtain
\[f_{\mathcal{X}}(x_3, x_2, x_4) = \frac{f_{\mathcal{X} \setminus \{x_1\}}(x_3, x_2, x_4)}{1 - p_1} = \frac{1}{1- p_1} \text{ and } f_{\mathcal{X}}(x_3, x_4, x_2) = \frac{f_{\mathcal{X} \setminus \{x_1\}}(x_3, x_4, x_2)}{1 - p_1} = 0.\]
Hence, 
$$\sum_{x \in \mathcal{X}} \left(f_{\mathcal{X}}(x, x_2, x_4) + f_{\mathcal{X}}(x, x_4, x_2) \right)p_x = p_1 + p_2 + \frac{p_3}{1-p_1} + p_4 \neq 1.$$

As the two other possible symmetries have been eliminated, we may now conclude that, for any $x, y \in \mathcal{X}$, 
\begin{equation} \label{eq:thm_symmetry} \sum_{w \in \mathcal{X}} (f_{\mathcal{X}}(x, w, y) + f_{\mathcal{X}}(y, w, x))p_w = 1.\end{equation}

Suppose now that $x, w, y \in \mathcal{X}$ are such that $d(x, w) > \min\{d(x, y), d(w, y)\}$.   Index the elements in $\mathcal{X}$ such that $d(x, z_n) > d(x, z_{n-1}) > ... > d(x, z_1)$.  Write $z_m = w$, $z_M = y$ and note that $z_1 = x$.  As above, let $\mathcal{X}_j = \{z_1, z_2, \dots, z_j\}$.  Note additionally that $x, w, y \in \mathcal{X}_{\max\{m, M\}}$, and thus $x, w, y \in \mathcal{X}_j$ for all $j \ge \max\{m, M\}$.  

Suppose that $j > \max\{m, M\}$.  Further, since $d(x, z_j) > d(x, z_i)$ for all $z_i \in \mathcal{X}_j \setminus \{z_j\}$, the pair $\{x, z_j\}$ is mutually outlying in $\mathcal{X}_j$ and $z_j \neq y$.  Now, by \eqref{eq:thm_fxsetminusz_summary}, it follows that
\begin{equation}\label{eq:thm_fxwyj} \frac{f_{\mathcal{X}_j}(x, w, y)}{m(\mathcal{X}_j)} =\left(\frac{1}{m(\mathcal{X}_j)}\right)\left(\frac{f_{\mathcal{X}_{j-1}}(x, w, y)}{1 - (p^{|{\mathcal{X}_j}})_{z_j}}\right) =  \frac{f_{\mathcal{X}_{j-1}}(x, w, y)}{m(\mathcal{X}_j)(1 - \frac{p_{z_j}}{m(\mathcal{X}_j)})} = \frac{f_{\mathcal{X}_{j-1}}(x, w, y)}{m(\mathcal{X}_{j-1})}.\end{equation}

Repeating the previous argument $n - \max\{m, M\}$ times, we obtain
\begin{equation} \label{eq:thm_maxmm} f_{\mathcal{X}}(x, w, y) = \frac{f_{\mathcal{X}_{\max\{m, M\}}}(x, w, y)}{m(\mathcal{X}_{\max\{m, M\}})}.\end{equation}

Consider first the case in which $m > M$.  First, as $z_m = w$, we have $d(x, w) > d(x, z_i)$ for all $1 \leq i < m$.  Therefore, $\{x, w\}$ is mutually outlying in $\mathcal{X}_{\max\{m, M\}} = \mathcal{X}_m$.  Hence, by \eqref{eq:thm_fxsetminusz_summary}, $f_{\mathcal{X}_{\max\{m, M\}}}(x, w, y) =0$, and so using
\eqref{eq:thm_maxmm}, we have $f_{\mathcal{X}}(x, w, y) = 0$.  

Now consider the case in which $M > m$.  First, as $z_M = y$, we have $d(x, y) > d(x, z_i)$ for all $1 \leq i < M$.  Therefore, $\{x, y\}$ is mutually outlying in $\mathcal{X}_{\max\{m, M\}} = \mathcal{X}_M$. Since $d(x, y) > d(x, w)$, using our original assumption that $d(x, w) > \min\{d(x, y), d(w, y)\}$, we conclude that $d(w, x) > d(w, y)$.  Hence, by \eqref{eq:thm_fxsetminusz_summary}, $f_{\mathcal{X}_{\max\{m, M\}}}(x, w, y) =0$, and so using
\eqref{eq:thm_maxmm}, we have $f_{\mathcal{X}}(x, w, y) = 0$.  In conclusion, if $x, w, y \in \mathcal{X}$ are such that $d(x, w) > \min\{d(x, y), d(w, y)\}$, we have $f_{\mathcal{X}}(x, w, y) = 0$.  
  
For ease of notation, for $x, w, y \in \mathcal{X}$ let us define
\begin{equation} \label{eq:thm_defout} \mathbf{1}_d(x, w, y) = \begin{cases} 1 & \text{ if } d(x, w) < \min\{d(x, y), d(w, y)\} \\ \alpha & \text{ if } x = w = y \\  0 & \text{ otherwise.}\end{cases}\end{equation} 
where $\alpha$ is some nonzero real number.  
By  \eqref{eq:thm_maxmm}, when $\mathbf{1}_d(x, w, y) = 0$ we have $f_{\mathcal{X}}(x, w, y) = 0$. Further, since provided that $y \neq x$, at most one of $\mathbf{1}_d(x, w, y)$ and $\mathbf{1}_d(y, w, x)$ are nonzero, there is a function $h_{\mathcal{X}}: \mathcal{X}^3 \rightarrow \mathbb{R}$ such that $h_{\mathcal{X}}(x, w, y) = h_{\mathcal{X}}(y, w, x)$ and
\begin{equation}\label{eq:thm_fh} f_{\mathcal{X}}(x, w, y) = \mathbf{1}_d(x, w, y) h_{\mathcal{X}}(x, w, y).\end{equation} Substituting this expression into \eqref{eq:thm_symmetry} and leveraging the symmetry of $h_{\mathcal{X}}$, for any $x,y \in \mathcal{X}$, 
$$\sum_{w \in \mathcal{X}} (\mathbf{1}_d(x, w, y) + \mathbf{1}_d(y, w, x))h_{\mathcal{X}}(x, w, y)p_w = 1.$$
For a fixed $(\mathcal{X}, d)$ and $x, y \in \mathcal{X}$,  $\mathbf{1}_d(x, w, y) + \mathbf{1}_d(y, w, x)$ is a constant which does not depend on the probability mass function, $p$.  Further, as the above equality must hold regardless of $p$ and for any fixed $x, y \in \mathcal{X}$, the function $h_{\mathcal{X}}(x, w, y)$ does not depend on $w$.  In particular, it must be the case that
$h_{\mathcal{X}}(x, w, y) = (\sum_{z \in \mathcal{X}}(\mathbf{1}_d(x, z, y) + \mathbf{1}_d(y, z, x))p_z)^{-1}$. Together with \eqref{eq:thm_fh}, we conclude that 
\begin{equation}\label{eq:thm_ff} f_{\mathcal{X}}(x, w, y) = \frac{\mathbf{1}_d(x, w, y)}{\sum_{z \in \mathcal{X}}(\mathbf{1}_d(x, z, y) + \mathbf{1}_d(y, z, x))p_z}.\end{equation}
If the denominator in the above expression were zero, since $\mathbf{1}_d(y, y, x) = 1$, it must be the case that $p_y = 0$, and thus, for the purposes of \eqref{eq:thm_gf}, we have $f_{\mathcal{X}}(x, w, y)p_y = 0$. Note that $\mathbf{1}_d(x, w, x) \neq 0$ if and only if $w = x$.  In such a case, the above expression is
$f_{\mathcal{X}}(x, x, x) = \frac{\alpha}{2 \alpha p_x} = \frac{1}{2 p_x}$, and thus does not depend on the choice of $\alpha$.  In particular, we can replace $\mathbf{1}_d(x, w, y)$ with the similarity comparison function, $\mathbb{1}_{d}(\{x, w\}, y)$, as in \Cref{def:Ind_Xd}.  
Pulling together \eqref{eq:thm_gf} and \eqref{eq:thm_ff}, we conclude $g_{\mathcal{X}}(x, w) = C_{\mathcal{X}}(x, w)$ as in \Cref{def:cohesion}.  
\end{proof}

\section{Future Directions}  As we have seen throughout, cohesion is a new measure of relative proximity that which allows highly concentrated, or point-like, sets to take on qualities of a single weighted point.  It is our hope that this self-contained exploration can facilitate the development of cohesion-based methods in exploratory data analysis (e.g., clustering, classification, low dimensional embedding, imputation) which can complement existing distance-based approaches.

\section{Appendix}\label{appendix}
\appendix

\begin{proof}[Proof of \Cref{lem:uxy}] Consider first the case that $x, y \in X_i$ for some $1 \leq i \leq n$.  Since $X_i$ is point-like, using also \Cref{def:subspace}, we have $$\mathbb{1}_{\mathcal{X}}(\{x, w\}, y) = \begin{cases}
\mathbb{1}_{X_i}(\{x, w\}, y) & \text{ if } w \in X_i \\
0 & \text{ if } w \in \mathcal{X} \setminus X_i. \end{cases}$$ Further, as the probability mass function for $X_i$ is $p|^{X_i}_x = \frac{p_x}{m(X_i)}$, it now follows that
\begin{eqnarray*} U_{\mathcal{X}}(x, y) &=& \sum_{w \in \mathcal{X}} \left(\mathbb{1}_{\mathcal{X}}(\{x, w\}, y) + \mathbb{1}_{\mathcal{X}}(\{y, w\}, x)\right)p_w\\
&=& m(X_i) \sum_{w \in X_i} \left(\mathbb{1}_{X_i}(\{x, w\}, y) + \mathbb{1}_{X_i}(\{y, w\}, x)\right)\frac{p_w}{m(X_i)}\\
&=& m(X_i) \sum_{w \in X_i} \left(\mathbb{1}_{X_i}(\{x, w\}, y) + \mathbb{1}_{X_i}(\{y, w\}, x)\right)p|^{X_i}_w\\
&=& m(X_i) U_{X_i}(x, y)\end{eqnarray*}

Consider now the case that $x \in X_i$ and $y \in X_j$ for some $i \neq j$.  
Given $w \in \mathcal{X}$,  $w \in X_k$ for some $1 \leq k \leq N$. By \cref{def:quotient}, $\mathbb{1}_{\mathcal{X}}(\{x, w\}, y) = \mathbb{1}_{\bar{\mathcal{X}}}(\{\bar{x}_i, \bar{x}_k\}, \bar{x}_j)$ and likewise $\mathbb{1}_{\mathcal{X}}(\{y, w\}, x) = \mathbb{1}_{\bar{\mathcal{X}}}(\{\bar{x}_j, \bar{x}_k\}, \bar{x}_i)$ and $\bar{p}_{\bar{x}_k} = \sum_{w \in X_k} p_w$. It now follows that 
\begin{align*} U_{\mathcal{X}}(x, y) &=&  \sum_{k = 1}^{N} \sum_{w \in X_k}  \left(\mathbb{1}_{\mathcal{X}}(\{x, w\}, y) + \mathbb{1}_{\mathcal{X}}(\{y, w\}, x)\right)p_w\\
&=& \sum_{k = 1}^{N}  \left(\mathbb{1}_{\bar{\mathcal{X}}}(\{\bar{x}_i, \bar{x}_k\}, \bar{x}_j) + \mathbb{1}_{\bar{\mathcal{X}}}(\{\bar{x}_j, \bar{x}_k\}, \bar{x}_i)\right)\sum_{w \in X_k} p_{w}\\
&=& \sum_{k = 1}^{N}  \left(\mathbb{1}_{\bar{\mathcal{X}}}(\{\bar{x}_i, \bar{x}_k\}, \bar{x}_j) + \mathbb{1}_{\bar{\mathcal{X}}}(\{\bar{x}_j, \bar{x}_k\}, \bar{x}_i)\right)\bar{p}_{\bar{x}_k}\\
&=& U_{\bar{\mathcal{X}}}(\bar{x}_i, \bar{x}_j).
\end{align*}

\end{proof}

\bibliographystyle{siamplain}
\bibliography{Moore_references}

 \end{document}